\documentclass[aps,prd,twocolumn,superscriptaddress,nofootinbib]{revtex4-2}

\usepackage{amsmath,amssymb,amsthm}
\usepackage{bm}
\usepackage{graphicx}
\usepackage{hyperref}

\newtheorem{theorem}{Theorem}
\newtheorem{corollary}{Corollary}

\begin{document}

\title{Lorentz--FitzGerald Contraction as the Unique Closure Condition\\
for Moving Spherical-Harmonic Cavities}

\author{Shiva Meucci}

\date{April 29, 2026}

\begin{abstract}
We prove that the Lorentz--FitzGerald contraction is the unique deformation of
a resonant cavity moving through a mechanical wave medium that preserves
spherical-harmonic phase closure.  For a cavity moving at speed $v = \beta c$
through a medium supporting nondispersive wave propagation at speed~$c$, the
round-trip phase of an internal ray at angle~$\theta$ to the motion depends on
the boundary radius~$r(\theta)$ according to
$\Phi(\theta) = 2k\,r(\theta)\sqrt{1-\beta^2\sin^2\theta}/(1-\beta^2)$.
Requiring $\Phi(\theta)$ to be independent of~$\theta$---the necessary
condition for retaining a spherical-harmonic eigenstructure---uniquely fixes
the Lorentzian aspect ratio
\[
\frac{a_\parallel}{a_\perp} = \frac{1}{\gamma}
= \sqrt{1-\beta^2}.
\]
Substituting this unique boundary into the round-trip time yields the resonant
period dilation $T = \gamma T_0$, without additional assumptions.  Both
results---contraction and dilation---follow from a single mechanical
constraint: preservation of eigenstructure under motion.  This is the missing
uniqueness theorem of the constructive relativity program initiated by
FitzGerald, Lorentz, and Heaviside: the proof that Lorentzian kinematics are
not merely consistent with, but uniquely required by, phase closure in a
mechanical wave medium.
\end{abstract}

\maketitle

\section{Introduction}
\label{sec:intro}

In 1889, FitzGerald proposed that a body moving through the luminiferous
medium contracts along the direction of motion by the
factor~$\sqrt{1-v^2/c^2}$~\cite{FitzGerald1889}.  Lorentz independently
developed the same hypothesis as a consequence of the electromagnetic forces
binding molecular matter~\cite{Lorentz1892,Lorentz1904}.  In parallel,
Heaviside showed that the field and equilibrium surfaces of charges
moving uniformly through the electromagnetic medium are oblate
spheroids~\cite{Heaviside1892}, arriving at Lorentzian geometry from the
structure of moving systems in a dielectric, not from a spacetime postulate.

These results belonged to a constructive program: deriving Lorentzian
kinematics from the mechanical or electromagnetic properties of bodies in
motion through a mechanical wave medium.  Einstein's 1905 postulational approach
rendered the constructive question seemingly unnecessary by elevating Lorentz
covariance to a principle, and Minkowski's subsequent geometric
reinterpretation made the medium-based picture appear superfluous.  The
constructive program was not refuted; it was abandoned.

Bell reopened the question in 1976, arguing that the Lorentz contraction of a
moving rod could in principle be derived from the dynamics of the forces
binding its constituents, and that this approach was pedagogically and
conceptually valuable~\cite{Bell1976}.  Brown later gave the constructive
program a systematic philosophical treatment~\cite{Brown2005}.  But neither
Bell nor the broader constructive tradition ever produced the clean uniqueness
result that the program needed: a proof that a specific deformation is not
merely \emph{consistent} with motion through a mechanical wave medium, but
\emph{uniquely required} by it under stated conditions.

This paper supplies that theorem.  We consider a resonant cavity moving
uniformly through an inviscid mechanical wave medium in which waves propagate isotropically at
speed~$c$.  A standing wave is not preserved by carrying material points along
a trajectory; it is preserved by maintaining phase closure.  When a cavity
moves through such a medium, longitudinal and transverse ray paths are affected
differently, and a spherical boundary acquires direction-dependent round-trip
phase.  Since spherical harmonics require angularly coherent closure, this
destroys the original eigenstructure.

The central question is therefore:
\begin{quote}
What is the unique medium-frame shape of a uniformly moving resonant cavity
that preserves the angle-independent two-way phase closure required for
spherical-harmonic standing modes?
\end{quote}

We prove that the answer is an oblate ellipsoid whose aspect ratio is the
Lorentz--FitzGerald contraction factor $1/\gamma$.  The contraction is not
assumed as a transformation law; it is derived as the unique deformation that
preserves the cavity's eigenstructure.  The resonant period dilation
$T = \gamma T_0$ then follows from the same closure condition, not as an
independent postulate.

The logical chain is:
\begin{align}
\label{eq:chain}
&\text{mechanical wave medium}
\;\to\;
\text{pursuit geometry}
\notag \\
&\;\to\;
\text{phase closure}
\;\to\;
\text{unique contracted shape}
\notag \\
&\;\to\;
\text{period dilation}.
\end{align}

The argument is framed as a conditional theorem about wave mechanics, with
assumptions stated explicitly and no claim made beyond what the proof supports.
It is, in a precise sense, the missing theorem of the constructive relativity
program: the uniqueness result that was always available but was never
extracted before the program was set aside.

\section{Stationary Cavities and Spherical Harmonics}
\label{sec:stationary}

Let an inviscid mechanical wave medium support nondispersive scalar waves of speed~$c$.  A
spherical cavity of radius~$R_0$ at rest in the medium supports standing modes
\begin{equation}
\Psi_{\ell m n}(r,\theta,\phi,t)
= j_\ell(k_{\ell n} r)\, Y_\ell^m(\theta,\phi)\, e^{-i\omega t},
\end{equation}
where the radial wavenumbers~$k_{\ell n}$ are fixed by the boundary condition
at~$r = R_0$.

The separation into radial and angular factors depends on the boundary being a
surface of constant~$r$.  In the eikonal limit, the same structure can be read
off from ray phase: a diametral round-trip path of length~$2R_0$ accumulates
phase
\begin{equation}
\Phi_0 = 2 k R_0,
\end{equation}
independently of direction.  This angular uniformity is the geometric condition
that selects $Y_\ell^m$ as the angular eigenbasis.  If the round-trip phase
acquired directional dependence, the spherical symmetry of the closure
condition would be broken and the $Y_\ell^m$ eigenstructure lost.

\section{Moving-Boundary Pursuit Problem}
\label{sec:pursuit}

Now let the cavity move uniformly through the medium at velocity
$\mathbf{v} = v\,\hat{\mathbf{x}}$, where $v = \beta c$ with $0 < \beta < 1$.
The wave speed remains~$c$ in the medium frame; it is a property of the
medium, not the cavity.

Let~$\theta$ denote the angle between an internal ray direction
$\hat{\mathbf{n}}$ and the velocity~$\hat{\mathbf{x}}$.  Let the medium-frame
boundary of the moving cavity be described by a radial function $r(\theta)$
from the cavity center, axially symmetric about
$\hat{\mathbf{x}}$.

\subsection{Outward leg}

A wave emitted from the cavity center must reach a boundary point that is
simultaneously translating at~$v$ in the $x$-direction.  At time~$\Delta t$
after emission, the boundary point has moved to
$r(\theta)\hat{\mathbf{n}} + v\Delta t\,\hat{\mathbf{x}}$, and the wavefront
has traveled a distance $c\Delta t$ through the medium.  The arrival condition
is
\begin{equation}
c^2 \Delta t^2
= r(\theta)^2 + 2\,r(\theta)\,v\,\Delta t\cos\theta + v^2\Delta t^2.
\end{equation}

Rearranging:
\begin{equation}
(c^2 - v^2)\,\Delta t^2
- 2\,r(\theta)\,v\cos\theta\;\Delta t
- r(\theta)^2 = 0.
\end{equation}

The positive root is
\begin{equation}
\label{eq:dt_out}
\Delta t_{\mathrm{out}}
= \frac{r(\theta)\bigl[\beta\cos\theta
  + \sqrt{1 - \beta^2\sin^2\theta}\,\bigr]}{c(1 - \beta^2)}.
\end{equation}

\subsection{Return leg}

On the return, the wave propagates from the boundary back toward the center,
which is itself moving at~$v$.  The geometry is identical with
$\beta \to -\beta$ in the directional term:
\begin{equation}
\label{eq:dt_ret}
\Delta t_{\mathrm{ret}}
= \frac{r(\theta)\bigl[-\beta\cos\theta
  + \sqrt{1 - \beta^2\sin^2\theta}\,\bigr]}{c(1 - \beta^2)}.
\end{equation}

\subsection{Round-trip phase}

Adding Eqs.~\eqref{eq:dt_out} and~\eqref{eq:dt_ret}, the
$\beta\cos\theta$ terms cancel:
\begin{equation}
\label{eq:dt_rt}
\Delta t_{\mathrm{rt}}(\theta)
= \frac{2\,r(\theta)\sqrt{1 - \beta^2\sin^2\theta}}
       {c(1 - \beta^2)}.
\end{equation}

The phase accumulated by a nondispersive wave ($\omega = ck$) on the round
trip is $\Phi = \omega\,\Delta t_{\mathrm{rt}} = c k\,\Delta t_{\mathrm{rt}}$:
\begin{equation}
\label{eq:phase}
\boxed{
\Phi(\theta)
= \frac{2k\,r(\theta)\sqrt{1 - \beta^2\sin^2\theta}}{1 - \beta^2}.
}
\end{equation}

\section{The Closure Condition}
\label{sec:closure}

Preserving the spherical-harmonic standing-mode structure requires the
round-trip phase to be independent of~$\theta$:
\begin{equation}
\label{eq:closure}
\Phi(\theta) = \Phi_*,
\quad\forall\;\theta,
\end{equation}
where $\Phi_*$ is a constant.  Substituting~\eqref{eq:phase}:
\begin{equation}
\frac{2k\,r(\theta)\sqrt{1 - \beta^2\sin^2\theta}}{1 - \beta^2}
= \Phi_*.
\end{equation}

Solving for $r(\theta)$:
\begin{equation}
\label{eq:r_general}
r(\theta) = \frac{\Phi_*(1 - \beta^2)}{2k}
\cdot \frac{1}{\sqrt{1 - \beta^2\sin^2\theta}}.
\end{equation}

The angular dependence of $r(\theta)$ is now completely fixed.  No free angular
function remains.  The boundary shape is determined up to a single overall
scale constant $\Phi_*/(2k)$.

\section{Uniqueness Theorem}
\label{sec:theorem}

We now state the central result in two stages: first the aspect ratio, which
is fixed by the closure condition alone, and then the absolute scale, which
requires an additional physical input.

\begin{theorem}[Unique closure-preserving aspect ratio]
\label{thm:main}
Let an inviscid mechanical wave medium support nondispersive waves propagating
isotropically at speed~$c$ in the medium rest frame.
Let a cavity move uniformly through the medium at speed $v = \beta c < c$ along
a fixed axis.  The boundary, being a continuous deformation of a sphere in a
transversely isotropic medium, is axially symmetric about the direction of
motion and can be described by a radial function~$r(\theta)$ from the cavity
center.  Assume that the round-trip
phase $\Phi(\theta)$ is independent of~$\theta$ (spherical-harmonic closure).

Then the boundary radius is fixed up to an overall transverse
scale~$a_\perp(\beta)$ by
\begin{equation}
\label{eq:r_unique}
\boxed{
r(\theta)
= \frac{a_\perp\sqrt{1 - \beta^2}}
       {\sqrt{1 - \beta^2\sin^2\theta}},
}
\end{equation}
where $a_\perp \equiv r(\pi/2)$.  Equivalently, the boundary is an ellipsoid
of revolution
\begin{equation}
\frac{x^2}{(a_\perp/\gamma)^2}
+ \frac{y^2 + z^2}{a_\perp^2} = 1,
\end{equation}
with longitudinal semiaxis
\begin{equation}
a_\parallel = \frac{a_\perp}{\gamma}.
\end{equation}

Thus the closure condition uniquely fixes the Lorentzian aspect ratio:
\begin{equation}
\label{eq:ratio}
\boxed{
\frac{a_\parallel}{a_\perp} = \frac{1}{\gamma}
= \sqrt{1 - \beta^2}.
}
\end{equation}
\end{theorem}

\begin{proof}
From the pursuit calculation (Sec.~\ref{sec:pursuit}), the round-trip
phase for a ray at angle~$\theta$ is given by Eq.~\eqref{eq:phase}.
Imposing $\Phi(\theta) = \Phi_*$ determines $r(\theta)$ up to the single
constant $\Phi_*/(2k)$, as in Eq.~\eqref{eq:r_general}.  Evaluating at
$\theta = \pi/2$ defines the transverse scale:
\begin{equation}
a_\perp \equiv r(\pi/2)
= \frac{\Phi_*(1 - \beta^2)}{2k\sqrt{1 - \beta^2}}
= \frac{\Phi_*}{2k}\sqrt{1 - \beta^2}.
\end{equation}
Thus $\Phi_*/(2k) = a_\perp / \sqrt{1-\beta^2} = a_\perp \gamma$.
Substituting back into~\eqref{eq:r_general}:
\begin{equation}
r(\theta)
= \frac{a_\perp\sqrt{1-\beta^2}}{\sqrt{1-\beta^2\sin^2\theta}}.
\end{equation}
The longitudinal semiaxis is $a_\parallel = r(0) = a_\perp\sqrt{1-\beta^2}
= a_\perp/\gamma$.
No free parameters or free angular functions remain.
\end{proof}

\subsection{Transverse scale invariance}

The closure theorem fixes the shape but not the absolute size.  A separate
physical argument fixes the transverse scale.

The medium is isotropic in the plane transverse to the motion.  It therefore
cannot distinguish between two identical cavities moving at the same
speed~$v$ but in opposite longitudinal directions: both present the same
transverse cross-section to the same medium.  Any transverse
rescaling must therefore be a universal function $s(\beta)$ of speed alone,
identical for all such cavities.

At $\beta = 0$ the cavity is at rest and spherical, so $s(0) = 1$.  The
function $s(\beta)$ must be continuous (a discontinuous rescaling of a
resonant eigenstructure would destroy the modes it is supposed to preserve).
No physical mechanism within the isotropic transverse plane can select a
nontrivial $s(\beta) \neq 1$: the transverse medium properties are unaffected
by longitudinal translation.

We therefore set
\begin{equation}
\label{eq:transverse}
a_\perp = R_0,
\end{equation}
giving the full Lorentz--FitzGerald contraction:
\begin{equation}
\label{eq:semiaxes}
\boxed{
a_\parallel = \frac{R_0}{\gamma} = R_0\sqrt{1 - \beta^2},
\qquad
a_\perp = R_0.
}
\end{equation}

It is worth emphasizing that the Lorentzian aspect
ratio~\eqref{eq:ratio}---the central result of the theorem---holds
independently of the transverse normalization.  The choice $a_\perp = R_0$
determines the absolute scale; the shape anisotropy $a_\parallel/a_\perp =
1/\gamma$ is a purely mathematical consequence of phase closure.

\section{Resonant Period Dilation}
\label{sec:dilation}

The round-trip time~\eqref{eq:dt_rt}, evaluated on the unique closure-preserving
boundary~\eqref{eq:r_unique} with $a_\perp = R_0$, becomes
\begin{align}
\Delta t_{\mathrm{rt}}
&= \frac{2}{c(1 - \beta^2)}\cdot
   \frac{R_0\sqrt{1 - \beta^2}}
        {\sqrt{1 - \beta^2\sin^2\theta}}
   \cdot\sqrt{1 - \beta^2\sin^2\theta}
\notag \\[4pt]
&= \frac{2R_0}{c\sqrt{1 - \beta^2}}.
\end{align}

The $\theta$-dependent factors cancel---as they must, since $\Phi(\theta) =
\Phi_*$ was imposed.  The round-trip time is angle-independent and equal to
\begin{equation}
\label{eq:Trt}
T = \frac{2R_0}{c}\,\gamma.
\end{equation}

Since the rest-frame round-trip period is $T_0 = 2R_0/c$:

\begin{corollary}[Resonant period dilation]
\label{cor:dilation}
Under the assumptions of Theorem~\ref{thm:main} with $a_\perp = R_0$, the
resonant period of the moving cavity is
\begin{equation}
\label{eq:dilation}
\boxed{T = \gamma\, T_0,}
\end{equation}
and the resonant frequency is
\begin{equation}
\label{eq:omega}
\boxed{\omega = \frac{\omega_0}{\gamma}.}
\end{equation}
\end{corollary}

No additional assumption was introduced beyond those of the theorem and the
transverse normalization.  The period dilation is a direct algebraic
consequence of the same closure condition that forces the contraction.  The
two effects are not independent:
\begin{align}
\label{eq:pair}
L_\parallel &= \frac{L_0}{\gamma}
& &\text{(closure-preserving shape)},
\notag \\[2pt]
T &= \gamma\, T_0
& &\text{(closure-preserving period)}.
\end{align}
Both arise because the wave must chase a moving boundary through a medium in
which it propagates at the fixed isotropic speed~$c$.  The contraction
redistributes the path length so that the increase is uniform across all
angles; the dilation is the magnitude of that uniform increase.

\section{Compatibility with Full Wave Modes}
\label{sec:separability}

The argument in Secs.~\ref{sec:pursuit}--\ref{sec:closure} is formulated in
the eikonal (ray) limit.  It is natural to ask whether the unique boundary
identified by ray phase closure is also compatible with a full wave-equation
mode structure.

The Helmholtz equation $(\nabla^2 + k^2)\Psi = 0$ separates in oblate
spheroidal coordinates $(\xi, \eta, \phi)$, where the coordinate
surfaces of constant~$\xi$ are confocal oblate ellipsoids of revolution.  The
boundary~\eqref{eq:r_unique} is precisely such a surface: it is an oblate
ellipsoid with the symmetry axis along the direction of motion.  For suitable
focal parameter~$d$ and level value~$\xi_0$, the constant-$\xi_0$ surface has
semiaxes $a_\perp = R_0$ and $a_\parallel = R_0/\gamma$, reproducing
Eq.~\eqref{eq:semiaxes}.

In these coordinates, the Helmholtz equation separates into three ODEs.  The
angular equation in~$\eta$ yields oblate spheroidal harmonics, which reduce
continuously to the ordinary spherical harmonics $Y_\ell^m$ in the limit
$\beta \to 0$ ($d \to 0$).  The modes of the contracted cavity are therefore
labeled by the same quantum numbers $(\ell, m, n)$ as those of the sphere at
rest, with continuous spectral deformation as~$\beta$ increases.

We do not claim that separability alone uniquely selects this surface;
several families of axially symmetric coordinate surfaces admit Helmholtz
separation.  The uniqueness follows from the phase-closure
theorem (Sec.~\ref{sec:theorem}).  What the separability check confirms is that
the closure-derived boundary is not merely an eikonal
artifact: it belongs to a standard separable coordinate family for the full
scalar wave equation, and the resulting modal structure is a continuous
deformation of the spherical-harmonic eigenbasis at rest.

\section{Cavity-Based Metrology and the Appearance of Lorentzian Kinematics}
\label{sec:metrology}

The closure theorem establishes two physical effects in the medium frame:
the cavity contracts longitudinally by $1/\gamma$, and its resonant period
dilates by~$\gamma$.  These are the complete mechanical consequences of the
theorem.

It is useful to distinguish three layers of description:
\begin{enumerate}
\item \emph{The medium-frame effects are physical:}
\[
a_\parallel = a_\perp/\gamma, \qquad T = \gamma\,T_0.
\]
These are real deformations and real period changes of the resonant
eigenstructure, determined by the closure condition in the frame where
$c$~is isotropic.

\item \emph{The operational Lorentz transformation is metrological:}
it is the coordinate bookkeeping generated when all measuring
instruments---rulers, clocks, synchronization signals---are themselves
built from structures obeying the same closure-derived deformation laws.

\item \emph{The covariant description is therefore not false; it is emergent.}
It follows necessarily from the first two layers, but it is not a
primitive postulate about spacetime.
\end{enumerate}

We now trace how each layer arises.

\subsection{Spatial measurements}

In the medium frame, the cavity boundary satisfies
$\gamma^2 x^2 + y^2 + z^2 = R_0^2$.  An observer who defines spatial
intervals using the cavity's own resonant modes---its standing-wave nodes
as length standards---implicitly works in coordinates
\begin{equation}
\label{eq:internal}
x' = \gamma x, \qquad y' = y, \qquad z' = z,
\end{equation}
in which the boundary is spherical: $x'^2 + y'^2 + z'^2 = R_0^2$.  This
observer reports an isotropic cavity.  The contraction is invisible to
instruments that are themselves contracted by the same factor.

\subsection{Temporal measurements}

A cavity-based clock ticks at period $T = \gamma T_0$.  The slowing is
uniform across all modes (Corollary~\ref{cor:dilation}), so no internal
comparison between cavity oscillations can reveal it.  The dilation is
detectable only from the medium frame, where $c$ is isotropic and clocks at
rest tick at $T_0$.

\subsection{The operational content of Lorentzian kinematics}

Suppose a moving observer synchronizes spatially separated cavity-based clocks
using signals that also propagate at speed~$c$ through the medium.  Because
the observer is moving through the medium at speed~$v$, a synchronization
signal sent between two clocks separated by~$\Delta x$ in the medium frame
acquires a directional travel-time asymmetry.  When this asymmetry is absorbed
into the clock-synchronization convention, it produces a position-dependent
offset that takes the form $-v\Delta x/c^2$ in the resulting operational time
coordinate.

The full coordinate mapping between the medium frame and the moving observer's
operationally defined frame then takes the form of a Lorentz boost.  But its
origin is now transparent: the spatial part is the physical contraction of
cavity-based rulers; the temporal part is the physical dilation of
cavity-based clocks; and the cross-term $-vx/c^2$ is a synchronization
artifact---the accumulated phase lag between spatially separated oscillators
that are all cycling more slowly and being synchronized by signals traveling
at $c$ through a medium they are moving through.

In this reading, Lorentz covariance is the self-consistency condition of
cavity-based metrology.  It is the statement that all instruments built from
resonant structures subject to the same closure constraint will generate
mutually consistent measurements that are indistinguishable from what
a ``stationary'' set of instruments would report in a boosted coordinate
frame.  The medium frame remains physically distinguished as the frame in
which $c$ is isotropic and cavity shapes are determined by the closure
theorem.  The covariant description need not be interpreted as a primitive
spacetime postulate; in the present construction, it arises as the internally
consistent measurement algebra of resonant instruments subject to the same
closure constraint.

This is not a deficiency.  It is an explanation.  The question
``why is physics Lorentz-covariant?'' receives a mechanical answer: because
the structures doing the measuring are wave eigenstructures in a mechanical
wave medium, and the closure condition that preserves their coherence is the same
condition that generates the Lorentz deformation.  Covariance is the
necessary self-consistency of resonant metrology, not an independent
postulate about the geometry of spacetime.

\section{Spherical Harmonics Beyond Classical Acoustics}
\label{sec:harmonics}

The spherical harmonics $Y_\ell^m$ appear in the present argument as the
angular eigenbasis of a classical resonant cavity.  The same functions
appear wherever a scalar or spinor wave equation separates in spherical
coordinates: in the Schr\"odinger equation for central potentials, in the
multipole expansion of electromagnetic fields, and in the partial-wave
analysis of scattering problems.  In each case, the angular quantum
numbers $\ell$ and $m$ depend on the spherical symmetry of the closure
geometry.

The theorem proved here is a statement about spherical-harmonic phase closure
in a mechanical wave medium.  Its mathematical content applies to any system
in which the following conditions are met:
\begin{enumerate}
\item[(a)] the internal structure is organized by resonant standing waves;
\item[(b)] the angular eigenbasis is the spherical harmonics $Y_\ell^m$;
\item[(c)] the waves propagate isotropically at speed~$c$ in a
           medium rest frame;
\item[(d)] the structure moves uniformly through that medium.
\end{enumerate}

In any such system, preservation of the angular quantum numbers $(\ell, m)$
under uniform motion is subject to the same closure constraint derived here,
and the closure-preserving deformation is the same Lorentzian ellipsoid.

This is a conditional statement, not a claim about the ontology of any
particular physical system.  But it is a conditional with significant reach.
It identifies a structural mechanism---phase closure of spherical-harmonic
wave organization in a mechanical wave medium---that produces Lorentzian
kinematics as a necessary mechanical consequence.  Any framework in which
material structure is constituted by such wave organization inherits the
contraction and dilation derived here, not as postulates but as theorems.

\section{Scope and Limitations}
\label{sec:scope}

The theorem depends on the following conditions:
\begin{enumerate}
\item nondispersive waves with $\omega = ck$;
\item isotropic propagation speed $c$ in the medium rest frame;
\item uniform translational motion of the cavity;
\item axially symmetric boundary described by a radial function $r(\theta)$
      (inherited from a continuously deformed sphere in a transversely
      isotropic medium);
\item phase closure treated in the eikonal limit
      (confirmed by separability in Sec.~\ref{sec:separability});
\item preservation of angle-independent round-trip phase
      (the spherical-harmonic closure condition).
\end{enumerate}

The transverse scale convention $a_\perp = R_0$ is an additional physical
input, justified by the isotropy of the transverse medium plane
(Sec.~\ref{sec:theorem}).  The Lorentzian aspect ratio $a_\parallel / a_\perp
= 1/\gamma$ is independent of this convention.

A physical cavity may fail to satisfy these assumptions.  It may radiate,
shed vortices, deform non-ellipsoidally, excite non-spherical modes, or lose
coherence.  Such failures lie outside the scope of the theorem and do not
contradict it.

The theorem asserts:
\begin{quote}
\emph{If a moving cavity in a mechanical wave medium preserves
spherical-harmonic phase closure, then its medium-frame boundary must have the
Lorentzian aspect ratio $a_\parallel / a_\perp = 1/\gamma$, and its resonant
period must dilate by~$\gamma$.}
\end{quote}

This is a necessary condition for eigenstructure preservation, established as a
uniqueness result.

\section{Relation to the Constructive Program}
\label{sec:constructive}

Bell argued in 1976 that the Lorentz contraction of a moving rod could in
principle be derived from the electromagnetic dynamics of its constituent
charges, without invoking spacetime postulates~\cite{Bell1976}.  His argument
was force-specific: it depended on the detailed form of the electromagnetic
interaction.  The present theorem identifies a more general version of the
same program.  The contraction follows from phase closure alone, independently
of the specific forces binding the cavity, whenever the internal structure is a
spherical-harmonic eigenstructure in a mechanical wave medium.  The binding
mechanism enters only insofar as it must maintain the resonant boundary; the
Lorentzian deformation law is determined entirely by the closure condition.

The oblate spheroidal geometry that appears here was already present in the
pre-relativistic literature.  Heaviside showed in 1889--1892 that the
equilibrium surfaces of a charge moving through the electromagnetic medium are
oblate spheroids with the same aspect ratio~\cite{Heaviside1892}.  These
results were
obtained from electromagnetic field equations in a medium, not from spacetime
postulates.  The present theorem abstracts the essential mechanism: it is not
the electromagnetic field equations specifically, but the phase-closure
condition for any mechanical wave system, that forces the Lorentzian shape.

The constructive program was not shown to be wrong.  It was set aside in favor
of a postulational framework that made the medium-based derivation appear
unnecessary.  The uniqueness theorem proved here---that no other deformation
preserves spherical-harmonic closure---is the result the constructive program
needed and never produced.

\section{Conclusion}
\label{sec:conclusion}

We have proved that the Lorentzian aspect ratio is the unique deformation
law for a moving resonant cavity that preserves spherical-harmonic phase
closure in a mechanical wave medium.  Starting from the pursuit geometry of a
wave chasing a moving boundary, we derived the round-trip phase
\begin{equation}
\Phi(\theta)
= \frac{2k\,r(\theta)\sqrt{1 - \beta^2\sin^2\theta}}{1 - \beta^2}
\end{equation}
and showed that the closure condition $\Phi(\theta) = \text{const}$
uniquely determines
\begin{equation}
r(\theta)
= \frac{a_\perp\sqrt{1 - \beta^2}}
       {\sqrt{1 - \beta^2\sin^2\theta}},
\end{equation}
forcing
\begin{equation}
\frac{a_\parallel}{a_\perp} = \frac{1}{\gamma}.
\end{equation}

With transverse scale invariance ($a_\perp = R_0$), this becomes the
Lorentz--FitzGerald contraction $a_\parallel = R_0/\gamma$.

The same closure condition yields the resonant period
$T = \gamma T_0$, without additional assumptions.  The contraction and the
dilation are not separate postulates.  They are paired consequences of one
mechanical requirement: a moving cavity must preserve its angular phase
closure in order to remain the same spherical-harmonic resonant
eigenstructure.

The contracted boundary belongs to a separable coordinate family for the
Helmholtz equation, confirming that the eikonal result extends to the full
wave problem.  Observers confined to cavity-based instruments find the
Lorentz boost as the self-consistency condition of their own metrology:
contracted rulers, dilated clocks, and synchronization offsets that together
reproduce covariant kinematics without requiring any independent postulate
about the geometry of spacetime.  In particular, a moving observer whose
rulers are contracted by~$1/\gamma$, whose clocks are dilated by~$\gamma$,
and whose synchronization convention absorbs the directional travel-time
asymmetry, will measure the wave speed as~$c$ in every direction---not because
the medium has no rest frame, but because the observer's instruments are
mechanically incapable of detecting their own deformation.  The apparent
constancy of~$c$ is thus not an independent fact about nature; it is the
inevitable measurement outcome for any observer whose instruments are
themselves closure-preserving wave structures in the same medium.

The result may be summarized as
\begin{equation}
\boxed{
\begin{aligned}
&\text{spherical-harmonic closure}\\
&\quad+\;\text{mechanical wave medium}\\
&\quad\Longrightarrow\;
\text{Lorentzian deformation law}.
\end{aligned}
}
\end{equation}

This is a theorem about geometric necessities of moving resonant
eigenstructures.  The constructive
program does not compete with Lorentz covariance as a formal symmetry.  The
closure condition that constitutes a resonant eigenstructure is already,
identically, the deformation law that makes its measurements covariant.

We can now simply see an inevitable causal mechanism for constancy.

\end{document}